\documentclass[10pt]{article}
\usepackage{amsmath, a4wide, amssymb, mathrsfs, bbm, chicago, amsthm}
\usepackage[english]{babel}
\usepackage[mathcal]{eucal}
\usepackage[title]{appendix}

\theoremstyle{plain} \newtheorem{thm}{Theorem}[section]
\theoremstyle{plain} \newtheorem{lem}[thm]{Lemma}
\theoremstyle{plain} 
\theoremstyle{plain} \newtheorem{prop}[thm]{Proposition}
\theoremstyle{plain} \newtheorem{cor}[thm]{Corollary}
\theoremstyle{definition} \newtheorem{defn}[thm]{Definition}

\newcommand{\real}[1]{\ensuremath{\mathbb{R}^{#1}}}

\title{{\bfseries Generalized Noether Theorems for Field Theories Formulated in Minkowski Spacetime}}
\author{M. Holman \\ Department of Physics and Astronomy, \hspace{0.05cm} Utrecht University, \\ Princetonplein 5,  \hspace{0.05cm} 3584 CC Utrecht, \hspace{0.05cm} The Netherlands \\ e-mail : {\ttfamily m.holman@phys.uu.nl}}

\begin{document}
\maketitle 
\begin{abstract}
\noindent New symmetry theorems are obtained for field theories formulated in Minkowski spacetime, based on the recognition 
that such theories should be diffeomorphism invariant. These theorems, which are in fact generalized Noether theorems, have
several nontrivial ramifications. 
One immediate consequence is the necessity to re-evaluate some of the default characterizations of the original Noether 
theorems as \emph{general statements} - for instance, the statement that every global invariance of the action gives rise 
to a conserved current.
As it turns out, this latter statement remains valid, but apparently only as a consequence of \emph{both} the generalized 
first and second Noether theorem (at any rate, it is very far from clear how to avoid such a conclusion).
A ramification of a more constructive nature consists in a novel expression for the so-called Belinfante stress-energy 
tensor, for which it is evident that it generalizes to the correct stress-energy tensor within the context of general 
relativity. The crucial symmetry property here is diffeomorphism invariance, rather than Poincar\'e invariance.
Using this new form for the stress-energy tensor, an alternative route to general relativity also becomes available, which, 
in contrast to many standard treatments, is intrinsically based on a variational formulation - i.e. without the need to
have prior knowledge of the form of Einstein's equation.
The connection of these results to related results previously established in the literature (although in a completely
different context) will be pointed out where applicable.
\end{abstract}

\section{Introduction}\label{intro}

\noindent Among the notable achievements in 20th century mathematical physics are undoubtedly the two symmetry theorems
put forward by \citeANP{Noether} in 1918. Roughly speaking (and properly paraphrased), Noether's first theorem states that
there exists a ``conserved current'' for every \emph{global} invariance of the action associated with a given physical theory,
while Noether's second theorem expresses the fact that for every \emph{local} invariance of the action, certain identities -
so-called generalized Bianchi identities - obtain. Classic examples illustrating the use of the first theorem are the conservation 
laws for energy, momentum and angular momentum, corresponding to the isometries of Minkowski spacetime, while the existence of 
certain constraints on the \emph{form} of any Lagrangian density involving gauge potentials and invariant under the corresponding 
local symmetries, can be viewed, essentially, as an illustration of the use of the second theorem.\\
Despite their common use as a standard tool in theoretical discourses however, there is something rather unsatisfactory about
both theorems in connection to external symmetries (which need not be ``spacetime symmetries''; more on this terminology in 
section \ref{Noetherthms}) within the context of Minkowski spacetime\footnote{Of course, within the context of general relativity, 
global external symmetries are generally absent, so that Noether's first theorem is in general not even applicable as far as 
such symmetries are concerned.}.
First, the components of the four conserved currents associated with the constant translations in space and time define 
a tensor of type $(0,2)$, $\Theta_{ab}$, the so-called canonical stress-energy-momentum  tensor (or canonical stress-energy
tensor, for short), but this tensor has unphysical properties in general.
That is, $\Theta_{ab}$ is not in general symmetric (so that it cannot be naturally generalized to a stress-energy tensor within the
context of general relativity), while for physically relevant theories involving gauge potentials, $\Theta_{ab}$ is not even
gauge invariant. To some extent, these problems can be remedied by noting that physical observables such as energy and momentum 
only correspond to the ``conserved charges'' of the corresponding Noether currents - i.e. to spatial integrals of the timelike
components of the Noether currents - and that there is moreover a way of modifying $\Theta_{ab}$ into a unique symmetric
tensor, which - under fairly general assumptions - gives rise to the same conserved charges\footnote{Cf. \citeANP{Belinfante1} \citeyear{Belinfante1,Belinfante2}.}.
As a matter of general principles (and logic) however, it is problematic to argue that a symmetry principle leads to a 
certain unique quantity, $X$, only to abandon this quantity in favour of another (still unique) quantity, $Y$, after discovering 
that $X$ has certain undesirable features. If the two quantities are ``physically equivalent'' in every respect
there is no need to introduce the quantity $Y$. If $Y$ actually is to be preferred for reasons of theoretical or empirical
consistency - as is the case for the stress-energy tensor - the connection with the symmetry principle becomes a priori unclear. As will be discussed
in subsection \ref{SEMtensor}, the resolution to this conceptual difficulty partly lies in the recognition that it is not meaningful
to consider constant translations in space and time in isolation, within a relativistic context. Yet, it will also be seen 
that this still leaves unanswered the question of the general interpretation of Noether's first theorem - i.e. essentially 
the statement that every independent, global invariance of the action gives rise to a conserved current - in terms of physics.\\
A problematic feature of Noether's second theorem derives from the fact that this theorem is formulated for \emph{arbitrary}
transformations of the dependent, \emph{as well as} the independent variables (i.e. for arbitrary transformations of the fields
on spacetime, as well as the spacetime coordinates) that leave invariant the action : if the expressions for the generalized
Bianchi identities commonly found in the literature are applied to Minkowski spacetime based field theories in the case of 
general coordinate transformations, inconsistencies are found to result\footnote{For such standard expressions, see e.g. 
\citeN{Utiyama}, or, more recently, \citeN{BradingBrown}.}. 
Now, it may appear to the reader that there is something illegitimate about asking for general coordinate 
invariance of field theories within the context of special relativity, but, in fact, no such thing is the case. Indeed, as is 
well known, \emph{any} physical theory formulated in terms of tensor fields defined on an arbitrary spacetime manifold, $M$, 
is form invariant with respect to arbitrary diffeomorphisms of $M$ (something which may be regarded, essentially, as a statement
about the general covariance of such theories\footnote{For a clear discussion of these points, see \citeN{Wald}, in particular section 4.1.}).
Regarding the apparent inapplicability of Noether's second theorem in the case of diffeomorphism invariance as pertaining to generic
(tensor) field theories defined in Minkowski spacetime, it thus seems that something has been overlooked in the standard 
discussions in presenting the most generic form of the generalized Bianchi identities.
Although it is not difficult to identify some possible causes in this respect (see sections \ref{Noetherthms} and \ref{comparison}
for some further remarks on this point), the key issue, as will be seen, is that it is (tacitly) assumed in these discussions
that all the dependent variables are also \emph{dynamical} variables. This means in particular that, when considered in 
complete generality, \emph{both} Noether theorems are actually inapplicable within the context of Minkowski spacetime.
Although this point has been noted in the literature\footnote{Cf. \citeN{Trautman}. See also \citeN{IyerWald}. In the main
part of the latter treatment a technical condition of ``diffeomorphism covariance'' - not to be confused with the 
condition of general covariance as characterized above - is imposed, which automatically excludes the presence of non-dynamical 
fields in the action.}, its consequences do not appear to have
been worked out so far, nor does it seem to be a generally appreciated fact in discussions on the subject.\\
In section \ref{Noetherthms}, derivations will be presented of (i) suitably generalized versions of Noether's two main 
theorems within the context of Minkowski spacetime and (ii) corresponding generalized expressions for the ``improper currents'' 
associated with general local symmetries. 
For Minkowski spacetime isometries, the results obtained here fully reduce to the familiar formulations of Noether's theorems,
something which is discussed in detail in section \ref{comparison}. As also discussed in that section however, there are several reasons, 
both empirical and theoretical, to attribute significance also to non-isometries in the case of Minkowski spacetime. Furthermore, 
the generalized Noether theorems derived here are in principle also relevant for alternative theories of gravity that involve 
``prior geometry'', such as Rosen's \citeyear{Rosen} bimetric theory. Although such alternatives go against the current paradigm 
that any viable theory of gravitation should be ``background independent'', the present experimental situation still seems to
leave some room for these theories\footnote{Cf. \citeN{Will}.} and the formulation of generalized Noether symmetry theorems may therefore 
not be a completely academic exercise in this case.\\
At any rate, in the present article attention will be restricted to field theories formulated in Minkowski spacetime and
it will be seen in section \ref{applications} that already in this case the generalized symmetry theorems have several nontrivial
ramifications. Perhaps the most suprising of these, is that the previously mentioned difficulty of associating a suitable 
stress-energy tensor with the constant translations in space and time via Noether's original theorem for \emph{global} symmetries, 
is resolved in a very natural way. The identically vanishing expression involving the (undifferentiated) ``improper currents'' 
associated with general \emph{local} symmetries, when applied to arbitrary spacetime diffeomorphisms, 
naturally decomposes into a symmetric stress-energy tensor, $T_{ab}$, which is \emph{manifestly} the flat spacetime 
``specialization'' of the general relativistic stress-energy tensor (i.e. defined as a functional derivative 
of the matter action with respect to the curved spacetime metric), and an expression which is the sum of the canonical 
stress-energy tensor and two additional terms.
Since the unique stress-energy tensor that results from adding specific improvement terms to the canonical stress-energy
tensor via the Belinfante method, when generalized to curved spacetime via the ``minimal substitution rule'', is known to 
agree with the general relativistic stress-energy tensor, one expects the sum of the two additional terms to be equivalent 
to the Belinfante improvement term and this expectation will indeed be proven correct in subsection \ref{SEMtensor}. 
As a result of this, the standard proof that the generalized Belinfante tensor agrees with the general relativistic 
stress-energy tensor shortens considerably. It is important to note however that, in contrast to the original Belinfante
method, the key symmetry property in the flat spacetime case is general diffeomorphism invariance, rather than 
Poincar\'e invariance.\\ 
A second nontrivial consequence of the generalized Noether theorems builds on the results of subsection \ref{SEMtensor} and
consists of an alternative route to Einstein's equation. In many standard treatments this equation is derived from 
various reasonable assumptions first and a variational formulation for general relativity is introduced only at a secondary
stage. Consistency then requires that the stress-energy tensor appearing in Einstein's equation be \emph{defined} as the
functional derivative of the matter action with respect to the metric. By contrast, if the alternative form for the 
stress-energy tensor is taken as a starting point, Einstein's equation can be derived \emph{intrinsically} from a 
variational formulation (together with various reasonable assumptions) - i.e. without the need to have prior knowledge of 
what it should be. Although both routes to general relativity are of course equally valid, they are conceptually quite 
distinct.\\
There is a certain non-uniformity in the notation for tensors employed throughout the present article. This is a result of
the circumstance that, on the one hand, general statements involving tensors are best expressed in terms of abstract 
indices\footnote{See e.g. \citeN{PenRin} for an exposition.} (here represented by lowercase latin letters), while, 
on the other hand, many of the results obtained in the following sections explicitly involve the introduction of coordinates,
and are therefore best expressed in terms of component indices (here represented by lowercase greek letters).
The reader who is confused by this notation may simply substitute abstract indices $a,b,c,\cdots$, wherever they appear,
by component indices $\mu, \nu, \kappa, \cdots$ and view the accompanying tensor equations as equations for the tensor
components with respect to a coordinate basis. Further (notational) conventions include bold face notation for differential
forms and the $-+++$ convention for the signature of the spacetime metric.

\section{Generalized Noether Theorems}\label{Noetherthms}

\noindent Recall that the integral of a function $F$ over an orientable spacetime manifold, $M$, is defined in
terms of the integration of (measurable) $4$-form fields. That is, by definition of orientability, there exists a volume
form on $M$, i.e. a continuous, nowhere vanishing $4$-form field, $\boldsymbol{\epsilon}$, and the integral of $F$ over $M$
(with respect to $\boldsymbol{\epsilon}$) is defined as $\int_M \! F \boldsymbol{\epsilon}$, where the integral of
the $4$-form field $F \boldsymbol{\epsilon}$ over $M$ is defined in the standard manner\footnote{See e.g. \citeN{AbMaRa}.}.
The arbitrariness in this definition that results from the non-uniqueness of $\boldsymbol{\epsilon}$ can be resolved, because the spacetime
metric, $g_{ab}$, can be used to impose the condition $\epsilon_{abcd} \epsilon^{abcd} = - 4 !$, which fixes $\boldsymbol{\epsilon}$
(up to overall orientation) to locally take the form $\boldsymbol{\epsilon} = \sqrt{-g} \, dx^0 \wedge dx^1 \wedge dx^2 \wedge dx^3$,
where $g = \det (g_{\mu \nu})$, $g_{\mu \nu}$ denoting the local components of the metric tensor (the $4$-form
$dx^0 \wedge dx^1 \wedge dx^2 \wedge dx^3$ represents the standard volume form associated with the dual basis field, 
$d x^{\mu}$ defined by the local coordinates $\{ x^{\mu} \}$). In order to define an action 
integral for a field theory on $M$ that does not have a dependence on the metric field hidden in $\boldsymbol{\epsilon}$, 
it is customary to specify a fixed volume form, $\boldsymbol{e}$, on $M$ and to then define the action as an integral of 
an appropriate \emph{scalar density} over $M$ with respect to $\boldsymbol{e}$. This means that the 
Lagrangian density, $\mathscr{L}$, can be expressed in the form $\mathscr{L} = \sqrt{-g} \mathscr{F}$, where $\mathscr{F}$ 
is a true scalar quantity that locally depends on the relevant field variables and a finite number of their spacetime 
covariant derivatives\footnote{More precisely, $\mathscr{L}$ is a scalar density of weight $1$. For a definition of the 
notion of a general tensor density of weight $w$, see e.g. \citeN{MTW}. Also, in writing $\mathscr{L}$ in the given form, 
it is implicitly understood that when local coordinates are introduced, the fixed volume form $\boldsymbol{e}$ equals the 
standard volume form associated with these coordinates (otherwise it is neccesary to divide $\sqrt{-g}$ by the component of 
$\boldsymbol{e}$ with respect to the standard volume form).}. In other words, on an arbitrary spacetime manifold, $M$, the 
action, $S$, for any field (or set of fields) $\psi$, is of the generic form
\begin{equation}\label{actiongeneric}
S  [\psi ] \; = \; \int_M \! \sqrt{-g} \mathscr{F} \, \boldsymbol{e}
\end{equation}
which is easily shown to be invariant under general coordinate transformations.\\
Now, in Minkowski spacetime, $(\real{4} , \eta_{ab})$, the fact that the Lagrangian density is a scalar density is usually 
forgotten, because in this case a canonical set of global coordinate systems (corresponding to the inertial motions in Minkowski
spacetime) is available for which $\sqrt{- \eta} = 1$, and as long as only transformations between these coordinate systems
are considered, there is no difference in properties between a scalar density and a true scalar. Given a set of global inertial
coordinates, $\{ x^{\mu} \}$, the action for any field theory can then simply be expressed as
\begin{equation}\label{action1}
S  [\psi ] \; = \; \int_{\real{4}} \! d^4x \, \mathscr{L} 
\end{equation}
which is the standard form for any action within the context of special relativity (with $\mathscr{L}$ treated as a true 
scalar\footnote{Eq. (\ref{action1}) follows from
the general expression (\ref{actiongeneric}) and the standard definition for the integral of a 4-form field.
In any local coordinate neighbourhood, $O$, of a generic spacetime for which $\boldsymbol{e}$ equals the standard volume
form, the integral of $\mathscr{L} \boldsymbol{e} = \sqrt{-g} \mathscr{F} \boldsymbol{e}$ over $O$ is defined as
\begin{equation*}
\int_O \!  \sqrt{-g} \mathscr{F} \boldsymbol{e} \; := \; \int_{\varphi(O)} \! d^4 x \, \sqrt{-g} \mathscr{F}
\end{equation*}
where the right-hand side of this equation is a standard integral over the open subset $\varphi(O)$ of \real{4} defined
by the local chart~$\varphi$.}).
The expression (\ref{action1}) is still generally correct (i.e. for inertial and non-inertial coordinate systems alike),  
as long as it is recalled that $\mathscr{L}$ is really a scalar density. Under a general, orientation-preserving
coordinate transformation, $x^{\mu} \rightarrow x'^{\mu}$, $\mathscr{L}$ transforms according to
$\left. \mathscr{L} \right|_{x} \longrightarrow \left. \mathscr{L} \right|_{x'} = \det (\partial x / \partial x') \left. \mathscr{L} \right|_{x}$ and hence
\begin{equation}
S  [\psi ] \; \longrightarrow S'[\psi'] \; = \; \int_{\real{4}} \! d^4x' \, \det (\partial x / \partial x') \left. \mathscr{L} \right|_{x} \; = \; S [\psi]
\end{equation}
where the standard change of variables theorem was used to express $S'$ as an integral with respect to the original
coordinates, to obtain the last equality. In fact, the above argument is valid in any local coordinate neighbourhood, which
proves the diffeomorphism invariance of the general action (\ref{actiongeneric}).
The reason for elaborating upon these elementary facts in some detail, is that the twofold appearance of the Jacobian
determinant associated with a generic coordinate transformation in any action integral of the form (\ref{action1}), with
the two contributions cancelling exactly, appears to have been ignored in some standard treatments\footnote{Cf. \citeN{Utiyama}, 
eq. (1.7) and the equation preceding it.}.\\
Having noted these facts, one can now proceed and derive generalized versions of the two Noether theorems.
In what follows, a generic, Lie algebra valued, dynamical tensor field, $\psi$, on Minkowski spacetime, will alternatively be denoted as
$\psi^i$ (i.e. with all indices - whether referring to spacetime or to an ``internal'' Lie algebra - represented by one upper
index $i$, which, moreover, also enumerates the various fields present). 
For present purposes, it is sufficient to assume that $\mathscr{L}$ is a local density that depends only on these fields,
their first spacetime derivatives and the nondynamical, flat background metric\footnote{The restrictions that spacetime is 
globally flat and that no second- or higher-order derivatives of the fields appear in $\mathscr{L}$ are imposed here only 
for convenience and are of no fundamental significance (see e.g. \citeN{LeeWald} for a more general treatment).}.
\begin{defn}\label{defsymm}
Given a classical action (\ref{action1}) with a local Lagrangian density of the form just specified, an infinitesimal transformation
\begin{equation}\label{inftrans1}
x \rightarrow x' \; = \; x \: + \: \delta x \qquad \quad \psi^i (x) \rightarrow {\psi'}^i (x') \; = \; \psi^i (x) \: + \: \delta \psi^i (x)
\end{equation}
which leaves invariant $S$, is called an infinitesimal \emph{symmetry transformation}. 
This definition can obviously be extended to finite transformations and it is then natural to assume - as will indeed
be done in what follows - that the set of all (finite) symmetry transformations thereby obtained carries a Lie group 
structure. In the transformations (\ref{inftrans1}), $x$ and $x'$ are taken to represent different coordinates of the 
same point in Minkowski spacetime, so that $\delta \psi^i$ represents the difference between the fields ${\psi'}^i$ and 
$\psi^i$ evaluated at the same spacetime point.
\end{defn}
\noindent More generally, a symmetry transformation may be defined as a finite transformation generated by (\ref{inftrans1}) which
leaves unaffected the dynamical field equations associated with the action. This latter definition of symmetries (which allows 
the action to be changed by a boundary term, for instance) will not be necessary for what follows however and will therefore
not be adopted. As a complement to the variation $\delta \psi^i$ in (\ref{inftrans1}), one could also consider evaluating 
${\psi'}^i$ at the same coordinate value, $x$, in the new coordinate system (corresponding to a different point in Minkowski 
spacetime) and define the (total) variation
\begin{equation}
\bar{\delta} \psi^i (x) \; := \; {\psi'}^i (x) \: - \: \psi^i (x)
\end{equation}
This way, one obtains to first order
\begin{equation}\label{inftrans2}
\delta \psi^i \; = \; \bar{\delta} \psi^i \: + \: (\partial_{\mu}\psi^i) \delta x^{\mu}
\end{equation}
where all quantities are understood to be evaluated at $x$ (here and elsewhere in this article the summation convention for 
summing over contracted (multi-)indices is followed). For future reference it is useful to state the following definition.
\begin{defn}
Infinitesimal symmetry transformations for which $\delta x^{\mu} = 0$ are called (infinitesimal) \emph{internal symmetries},
whereas infinitesimal symmetry transformations induced entirely by coordinate transformations, $x \rightarrow x + \delta x$, are referred 
to as (infinitesimal) \emph{external symmetries}\footnote{Very often, external symmetries are alternatively referred to as 
``spacetime symmetries''. In the present work, spacetime symmetries always refer to symmetries \emph{of} spacetime which preserve 
at least its causal structure (and thus in particular preserve more than just the local manifold structure, which is all that
is preserved by generic external symmetries).}.
\end{defn}
\noindent It is now straightforward to proceed as follows. Let $\Omega$ denote a compact region of Minkowski spacetime.
Given \emph{any} infinitesimal transformation (\ref{inftrans1}) of both the dependent and the independent variables, the action
transforms into
\begin{eqnarray}
S'[\psi'] & = & \int_{\Omega'} \! d^4x' \, \det (\partial x / \partial x') \left. \mathscr{L} (\psi + \delta \psi , \partial \psi + \delta \partial \psi , \eta + \delta \eta) \right|_{x}  \nonumber \\
	  & := & S[\psi] \: + \: \delta S			\label{varaction0}
\end{eqnarray}
and it is easily checked that, to first order, $\delta S$ can be expressed as 
\begin{equation}\label{NoetherVP}
\delta S \; = \; \int_{\Omega} \! d^4x \, \left\{ \partial_{\mu} \left( \frac{\partial \mathscr{L}}{\partial (\partial_{\mu} \psi^i)} \delta \psi^i \: + \: \Theta^{\mu}_{\mspace{12mu} \nu} \delta x^{\nu} \right) \: + \: E_i(\psi) \bar{\delta} \psi^i 
\: - \: \left( \mathscr{L} \partial_{\mu} \delta x^{\mu} - \frac{\partial \mathscr{L}}{\partial \eta^{\mu \nu}} \delta \eta^{\mu \nu} \right) \right\}
\end{equation}
where the quantities $E_i$ represent the familiar Lagrange expressions 
\begin{equation}\label{Lagrexpr}
E_i (\psi)  \; := \; \frac{\partial \mathscr{L}}{\partial \psi^i} \: - \: \partial_{\mu} \frac{\partial \mathscr{L}}{\partial (\partial_{\mu} \psi^i)}          
\end{equation}
and where the 16 quantities $\Theta^{\mu}_{\mspace{12mu} \nu}$ are defined by
\begin{equation}\label{canonicalSEtensor}
\Theta^{\mu}_{\mspace{12mu} \nu} \; := \; - \, \frac{\partial \mathscr{L}}{\partial (\partial_{\mu} \psi^i)} \partial_{\nu} \psi^i  \; + \; \mathscr{L} \delta^{\mu}_{\mspace{12mu} \nu}
\end{equation}
and denote the components of the \emph{canonical stress-energy tensor} mentioned in the introduction\footnote{It should be
noted that there is a slight ambiguity in notation in that in eq. (\ref{varaction0}) and all expressions derived from it, 
$\mathscr{L}$ actually represents the scalar quantity (previously denoted by $\mathscr{F}$) which it equals in any global 
inertial coordinate system and an additional factor $\sqrt{- \eta}$ (evaluated at $x$) should in general be included in 
the integrand on the right-hand side of the first line of eq. (\ref{varaction0}). This should be kept in mind especially
when considering quantities involving derivatives with respect to $\eta$, such as in eq. (\ref{NoetherVP}) or in the definition
(\ref{stressenergytensor1}) in the main text. The reason for writing $\mathscr{L}$ instead of 
$\mathscr{F}$ is simply to facilitate comparison with standard treatments of the Noether theorems, as referred to earlier.
It is also important to note that differentiation with respect to $\eta$ is only defined for variations in $\eta$
towards other flat metrics and should therefore be understood symbolically.
Finally, it is noted that the sign in the definition (\ref{canonicalSEtensor}) of the canonical stress-energy tensor
is dictated by the use of the $-+++$ signature for the metric and the requirement that $\Theta_{00}$ equals the (positive) 
Hamiltonian density as determined from the standard canonical formalism.\label{notationremark}}.
Since the region $\Omega$ is further arbitrary, one thus obtains the following key result.
\begin{lem}[Solution to Noether's Variational Problem]\label{NoetherVPlem}
For an infinitesimal transformation (\ref{inftrans1}) to be a symmetry transformation, it is necessary that
\begin{equation}\label{NoetherVP1}
\partial_{\mu} \left( \frac{\partial \mathscr{L}}{\partial (\partial_{\mu} \psi^i)} \delta \psi^i \: + \: \Theta^{\mu}_{\mspace{12mu} \nu} \delta x^{\nu} \right) \: + \: E_i(\psi) \bar{\delta} \psi^i \: - \: \left( \mathscr{L} \partial_{\mu} \delta x^{\mu} - \frac{\partial \mathscr{L}}{\partial \eta^{\mu \nu}} \delta \eta^{\mu \nu} \right) \; = \; 0
\end{equation}
or, alternatively, 
\begin{equation}\label{NoetherVP2}
\partial_{\mu} \left( \frac{\partial \mathscr{L}}{\partial (\partial_{\mu} \psi^i)} \bar{\delta} \psi^i \: + \: \mathscr{L} \delta x^{\mu} \right) \; =  \; - E_i(\psi) \bar{\delta} \psi^i \: + \: \left( \mathscr{L} \partial_{\mu} \delta x^{\mu} - \frac{\partial \mathscr{L}}{\partial \eta^{\mu \nu}} \delta \eta^{\mu \nu} \right)
\end{equation}
\end{lem}
\noindent For a given transformation (\ref{inftrans1}) of specified type, either eq. (\ref{NoetherVP1}) or eq. (\ref{NoetherVP2})
expresses a condition on the form of the action, that must be satisfied if the transformation is to be a symmetry.
As such, these equations provide a solution to the ``variational problem'' with which Noether was concerned in her 1918
article. It is important to note however that the condition expressed by either eq. (\ref{NoetherVP1}) or eq. (\ref{NoetherVP2})
deviates from the expression for this condition commonly found in the literature\footnote{See e.g. \citeN{BradingBrown} and
further references therein.}. 
As already remarked in the introduction, there are a number of possible explanations for this disagreement.
For instance, the absence of the first term in round brackets on the right-hand side of eq. (\ref{NoetherVP2}) in the usual 
expressions can be explained by a failure to take into account that $\mathscr{L}$ is a scalar density, whereas the absence 
of both terms in round brackets on the right-hand side of eq. (\ref{NoetherVP2}) in the usual expressions is actually 
explainable as a consequence of a (tacit) assumption that it suffices to consider only \emph{isometries} in the case of Minkowski 
spacetime (as will be seen in section \ref{comparison} however, it is not in fact \emph{necessary} for a coordinate 
transformation of Minkowski spacetime to be an isometry, in order for the general solution to the variational problem 
to agree with the expression commonly found for it in the literature).
However, a closer inspection of the standard discussions reveals that these possible accounts are insufficient as 
genuine explanations of the disagreement. This is because the standard expressions for the condition (\ref{NoetherVP2}) 
are formulated (or at least claimed to be so) for arbitrary coordinate transformations of arbitrary spacetimes and these 
expressions should therefore in particular be valid in the case of arbitrary coordinate transformations of Minkowksi spacetime. 
In fact, as will become clear in the following sections, the root of the disagreement between condition (\ref{NoetherVP2})
and the standard expressions for it, lies in the tacit assumption within the standard treatments that all dependent variables
are actually dynamical variables.\\
Before deriving generalized Noether theorems from the general solution to Noether's variational problem in the case of 
Minkowski spacetime, it is useful to further differentiate between types of symmetry.
\begin{defn}
Symmetry transformations (\ref{inftrans1}) that generate a finite-dimensional Lie group and for which there exists no compact 
spacetime region outside which they are trivial (with the exception of the trivial transformation itself), are 
called \emph{global symmetries}\footnote{The reason for this (somewhat unusual) second condition is to exclude diffeomorphisms
generated by any nonzero vector field of compact support on $\real{4}$ from being labelled as global symmetries.
Global symmetries thus correspond to finite-dimensional Lie groups and are ``globally nontrivial''.}.
Symmetry transformations (\ref{inftrans1}) that generate a Lie group of non-finite dimension larger than $\aleph_0$,
are called \emph{local symmetries} or \emph{gauge symmetries}\footnote{In modern discourses, it is customary to use the terms ``gauge (symmetry) transformation'' and ``local 
symmetry transformation'' (either internal or external) interchangeably and this modern terminology is adopted here.
In particular, the phrase ``local gauge transformation'' is regarded as tautological. It should also be noted that local
symmetry groups are not Lie groups in the strict technical sense, but so-called Fr\'echet-Lie groups. Nevertheless, the standard convention 
will be followed here of understanding the term ``Lie group'' in a more generic sense that includes local symmetry groups.\label{gaugeterminology}}.
\end{defn}
\noindent If the action for a given physical theory is known to admit a certain group of global symmetries, then, by definition,
the infinitesimal transformations (\ref{inftrans1}) leaving $S$ invariant, can, to first order, be expressed according to
\begin{equation}\label{infvar1}
\delta x^{\mu} \; := \; X^{\mu}_{\mspace{12mu} \rho}(x,\psi) \epsilon^{\rho}   \qquad \quad   \delta \psi^i (x) \; := \; \Psi^i_{\mspace{8mu} \rho}(x,\psi) \epsilon^{\rho} 
\end{equation}
for certain functions $X^{\mu}_{\mspace{12mu} \rho}$, $\Psi^i_{\mspace{8mu} \rho}$ of $x$ and $\psi$ and for a certain,
finite number, $r$, of spacetime independent, arbitrary infinitesimal parameters, $\epsilon_{\rho}$.
In what follows, a generic, $r$-dimensional group of global symmetries will be denoted as $G_r$.
As is well known, the physical interactions that are presently regarded as fundamental are all described by actions that
exhibit a more general kind of invariance, involving arbitrary functions of the spacetime coordinates, rather than
spacetime independent parameters. That is, the actions in question are known to be invariant with respect to certain groups
of local symmetries and for many cases of practical interest the infinitesimal transformations (\ref{inftrans1}) leaving
$S$ invariant can be expressed as
\begin{equation}\label{infvar2}
\delta x^{\mu} \; := \; \bar{X}^{\mu}_{\mspace{12mu} \rho} \xi^{\rho}   \qquad \quad   \delta \psi^i \; := \; ^0 \Psi^i_{\mspace{8mu} \rho} \xi^{\rho} \: + \: ^1 \Psi^{i \mu}_{\mspace{16mu} \rho} \partial_{\mu} \xi^{\rho} 
\end{equation}
for a certain finite number, $r'$, of smooth, but further arbitrary, infinitesimal parameter functions, $\xi^{\rho}$, 
of the spacetime coordinates (in the same way as before, the coefficients $\bar{X}^{\mu}_{\mspace{12mu} \rho}$, $^0 \Psi^i_{\mspace{8mu} \rho}$, 
$^1 \Psi^{i \mu}_{\mspace{16mu} \rho}$, depend on both the spacetime coordinates and the field variables\footnote{There
is again no fundamental significance in considering the variations $\delta \psi^i$
to only depend on the $\xi^{\rho}$ and their first derivatives. See e.g. \citeN{AndBer} for a discussion of the general case.}).
In what follows, a generic group of local symmetries, involving $r'$ smooth parameter functions, will
be denoted as $G_{\infty r'}$. A subgroup of any local symmetry group that contains only global symmetries will be referred to as 
a \emph{global subgroup} of the local symmetry goup. The foregoing definitions naturally lead to the following proposition.
\begin{prop}
Let $G_{\infty r'}$ denote a local symmetry group for a classical field theory defined on Minkowski spacetime. Then
$G_{\infty r'}$ has a $r'$-dimensional global subgroup, $G_{r'}$.
\end{prop}
\noindent Although this may seem a trivial observation (i.e. simply take the infinitesimal parameters in (\ref{infvar2}) to
be constants), it is explicitly stated here, as claims to the opposite have been made in the recent literature\footnote{For a detailed discussion of this issue, see \citeN{Holman}.}.
More generally, a local symmetry group $G_{\infty r'}$ may have global subgroups larger than $G_{r'}$. For instance,
in the case where $G_{\infty r'}$ is the group, $\mbox{Diff}(\real{4})$, consisting of the diffeomorphisms of Minkowski 
spacetime, the largest global subgroup is known to be the $24$-dimensional group of fractional linear transformations
of $\real{4}$ into itself\footnote{Strictly speaking this of course requires a suitable extension of Minkowski spacetime
by adding some points at infinity to it. There are standard, mathematically well defined constructions to achieve this,
but which will be left implicit in what follows.}. 

\subsection{Generalized Noether Theorem for Global Invariances of the Action}

\noindent Let it be known that the action is invariant under a $r$-parameter Lie group, $G_r$, so that the infinitesimal
transformations that leave invariant $S$ can be expressed in the form (\ref{infvar1}). Then it is easily verified that the
expression (\ref{NoetherVP}) for $\delta S$ can be put into the form
\begin{equation}\label{varaction1}
\delta S \; = \; \int_{\Omega} \! d^4x \, \left( \partial_{\mu} J^{\mu}_{\mspace{12mu} \rho} \: + \: E_i(\psi) \gamma^i_{\mspace{12mu} \rho}
\: - \: T_{\mu \nu} \partial^{\mu} X^{\nu}_{\mspace{12mu} \rho} \right) \epsilon^{\rho}
\end{equation}
where the $r$ currents, $J^{\mu}_{\mspace{12mu} \rho}$, are given by
\begin{equation}\label{Noethercurrent}     
J^{\mu}_{\mspace{12mu} \rho} \; := \; \frac{\partial \mathscr{L}}{\partial (\partial_{\mu}\psi^i)} \Psi^i_{\mspace{8mu} \rho} \: + \: \Theta^{\mu}_{\mspace{12mu} \nu}X^{\nu}_{\mspace{12mu} \rho}  \qquad \quad \rho = 1, \cdots , r \; ,
\end{equation}
the $16$ quantities $T_{\mu \nu}$ are defined according to
\begin{equation}\label{stressenergytensor1}
T_{\mu \nu} \; := \; \mathscr{L} \eta_{\mu \nu}  - 2 \frac{\partial \mathscr{L}}{\partial \eta^{\mu \nu}} \qquad \quad \mu , \nu  = 0 ,1 ,2 ,3
\end{equation}
and where the coefficients $\gamma^i_{\mspace{12mu} \rho}$ are defined by $\gamma^i_{\mspace{12mu} \rho} := \Psi^i_{\mspace{8mu} \rho} \: - \: (\partial_{\nu} \psi^i)) \bar{X}^{\nu}_{\mspace{12mu} \rho}$.
Since this expression for $\delta S$ vanishes for arbitrary infinitesimal $\epsilon^{\rho}$, one thus arrives at the
following theorem, which is a generalized version of Noether's first theorem for the case of Minkowski spacetime
\begin{thm}\label{Ntheorem1}
If the action, $S$, is invariant with respect to a Lie group, $G_r$, of global symmetry transformations of the Minkowski spacetime 
coordinates and/or the fields, there are $r$ divergence relations of the form
\begin{equation}\label{Noetherthm1}
\partial_{\mu} J^{\mu}_{\mspace{12mu} \rho} \: + \: E_i (\psi) \gamma^i_{\mspace{12mu} \rho}  \: - \: T_{\mu \nu} \, \partial^{\mu} X^{\nu}_{\mspace{12mu} \rho} \; = \; 0
\end{equation}
\end{thm}
\noindent If attention is restricted to the subgroup of $G_r$ consisting of those transformations for which the coordinate
transformations, $x^{\mu} \rightarrow x^{\mu} + \delta x^{\mu}$, are isometries of Minkowski spacetime, the last term in 
eq. (\ref{Noetherthm1}), which is non-standard, vanishes, and Noether's original result for global symmetries is obtained.
For further discussion of the precise relation between theorem \ref{Ntheorem1} and Noether's original first theorem, see 
section \ref{comparison}.

\subsection{Generalized Noether Theorems for Local Invariances of the Action}

\noindent Let it be known that the action is invariant under an infinite-dimensional Lie group, $G_{\infty r'}$, so that the 
infinitesimal transformations that leave invariant $S$ can be expressed in the form (\ref{infvar2}). Then it is easily verified that the
expression (\ref{NoetherVP}) for $\delta S$ can be put into the form
\begin{equation}\label{varaction2}
\delta S \; = \; \int_{\Omega} \! d^4x \, \left( \partial_{\mu} B^{\mu} \: + \: \Xi_{\rho} \xi^{\rho} \right)
\end{equation}
where the ``currents'', $B^{\mu}$, are defined by
\begin{eqnarray}\label{Bcurrent}
B^{\mu} & := &  \frac{\partial \mathscr{L}}{\partial (\partial_{\mu} \psi^i)} \delta \psi^i \: + \: \Theta^{\mu}_{\mspace{12mu} \nu} \delta x^{\nu} \: + \: E_i(\psi) ^1 \Psi^{i \mu}_{\mspace{16mu} \rho} \xi^{\rho}  \: - \: \mathscr{L} \delta x^{\mu} \: + \: 2 \eta^{\alpha \mu} \frac{\partial \mathscr{L}}{\partial \eta^{\alpha \nu}} \delta x^{\nu}   \nonumber \\
	& = & \frac{\partial \mathscr{L}}{\partial (\partial_{\mu} \psi^i)} \bar{\delta} \psi^i \: + \: E_i(\psi) ^1 \Psi^{i \mu}_{\mspace{16mu} \rho} \xi^{\rho}  \: + \: 2 \eta^{\alpha \mu} \frac{\partial \mathscr{L}}{\partial \eta^{\alpha \nu}} \delta x^{\nu}   
\end{eqnarray}
and where the $r'$ quantities $\Xi_{\rho}$ are defined by
\begin{equation}\label{noetherlocal1}
\Xi_{\rho} \; := \; \bigl\{ E_i (\psi) ( ^0 \Psi^i_{\mspace{8mu} \rho} \: - \: \bar{X}^{\mu}_{\mspace{12mu} \rho}(\partial_{\mu} \psi^i)) \: - \: \partial_{\mu} (E_i (\psi) ^1 \Psi^{i \mu}_{\mspace{16mu} \rho}) \bigr\} \: + \: (\partial^{\mu}T_{\mu \nu}) \bar{X}^{\nu}_{\mspace{12mu} \rho}
\end{equation}
Since this expression for $\delta S$ vanishes for arbitrary infinitesimal parameter functions $\xi^{\rho}$, it vanishes also
for the specific such functions that vanish and have vanishing derivatives on $\partial \Omega$. Hence, for such variations
$\Xi_{\rho} = 0$, as the variations are further arbitrary within $\Omega$. But the expressions $\Xi_{\rho}$ are independent
of the variations and hence, the following generalized form of Noether's second theorem in Minkowski spacetime is obtained
\begin{thm}\label{Ntheorem2}
If the action, $S$, is invariant with respect to a Lie group, $G_{\infty r'}$, of local symmetry transformations of the Minkowski
spacetime coordinates and/or the fields, the following identities obtain, \emph{independent} of any field equations
\begin{equation}\label{Noetherthm2}
\Xi_{\rho} \; = \; 0 \qquad \quad \rho = 1 , \cdots , r'
\end{equation}
with the $\Xi_{\rho}$ given by eq. (\ref{noetherlocal1}).
\end{thm}
\noindent Again, the last term in the expression (\ref{noetherlocal1}) for the $\Xi_{\rho}$ is a non-standard term.
The original version of Noether's second theorem is often described as stating the existence of dependencies among
the Lagrange expressions - so-called generalized Bianchi identities - as a direct consequence of the gauge symmetry of the
action and independent of any field equations. It is thus seen that, in the case of Minkowski spacetime, the dependencies among the 
Lagrange expressions in general include an inhomogeneous term. Further discussion of the relation between theorem 
\ref{Ntheorem2} and the usual form of the second Noether theorem can be found in section \ref{comparison}.\\
Because of the identities (\ref{Noetherthm2}), the first term in eq. (\ref{varaction2}) must
vanish for arbitrary variations, which can only happen if each coefficient of the infinitesimal parameter functions and their
spacetime derivatives vanishes separately. It is straightforward to check that this implies the following result.
\begin{thm}\label{Ntheorem3}
If the action, $S$, is invariant with respect to a Lie group, $G_{\infty r'}$, of local symmetry transformations of the Minkowski
spacetime coordinates and/or the fields, the following $11r'$ identities obtain, \emph{independent} of any field equations
\begin{eqnarray}
\partial_{\mu} \mathscr{J}^{\mu}_{\mspace{12mu} \rho} & = & 0                                                                                                                                                                     \label{utiya1}  \\
\mathscr{J}^{\mu}_{\mspace{12mu} \rho} \: + \: \partial_{\nu} G^{\mu \nu}_{\mspace{18mu} \rho}  & = & 0                                                  \label{utiya2}  \\
G^{\mu \nu}_{\mspace{18mu} \rho} \: + \:  G^{\nu \mu}_{\mspace{18mu} \rho} & = & 0    \label{utiya3}
\end{eqnarray}
for $\rho = 1 , \cdots , r'$, $\mu , \nu = 0 , \cdots , 3$, where the $\mathscr{J}^{\mu}_{\mspace{12mu} \rho}$
denote the ``improper currents''
\begin{eqnarray}
\mathscr{J}^{\mu}_{\mspace{12mu} \rho} & := & E_i (\psi) ^1 \Psi^{i \mu}_{\mspace{16mu} \rho} \: + \: \Theta^{\mu}_{\mspace{12mu} \nu} \bar{X}^{\nu}_{\mspace{12mu} \rho} \: + \:  ^0 \Psi^i_{\mspace{8mu} \rho} \frac{\partial \mathscr{L}}{\partial (\partial_{\mu}\psi^i)}  \: - \: \mathscr{L} \bar{X}^{\mu}_{\mspace{12mu} \rho} \: + \: 2 \eta^{\alpha \mu} \frac{\partial \mathscr{L}}{\partial \eta^{\alpha \nu}} \bar{X}^{\nu}_{\mspace{12mu} \rho}  \nonumber \\
                     & = &  J^{\mu}_{\mspace{12mu} \rho} \: + \: E_i (\psi) ^1 \Psi^{i \mu}_{\mspace{16mu} \rho} \: - \: T^{\mu}_{\mspace{12mu} \nu} \bar{X}^{\nu}_{\mspace{12mu} \rho} \label{current1}
\end{eqnarray}
the $G^{\mu \nu}_{\mspace{18mu} \rho}$ are defined by
\begin{equation}\label{intcurrent1}
G^{\mu \nu}_{\mspace{18mu} \rho} \; := \; \, ^1 \Psi^{i \mu}_{\mspace{16mu} \rho} \frac{\partial \mathscr{L}}{\partial (\partial_{\nu} \psi^i)}   
\end{equation}
and the $J^{\mu}_{\mspace{12mu} \rho}$ denote the Noether currents (\ref{Noethercurrent}), corresponding to the
subgroup of global symmetries associated with the general symmetry variations (\ref{infvar2}) (i.e. with $X^{\nu}_{\mspace{12mu} \rho}$
and $\Psi^i_{\mspace{8mu} \rho}$ replaced by respectively $\bar{X}^{\nu}_{\mspace{12mu} \rho}$ and $^0 \Psi^i_{\mspace{8mu} \rho}$).
\end{thm}
\noindent Another way of formulating these results is that, amongst other things, the presence of $r'$ independent local 
symmetries implies the existence of $r'$ currents, which are conserved identically - independent of any field equations.
Before closing this section, it is worth noting that if a definition of symmetries is adopted in which the total change 
of the action, $\Delta S$, under a generic transformation (\ref{inftrans1}) is allowed to be of the form 
$\int_{\real{4}} \! d^4 x \, \partial_{\mu} Q^{\mu}$ (assuming $Q^{\mu}$ to be $C^1$; cf. the remarks in definition \ref{defsymm}),
it is necessary to include the divergence of $Q^{\mu}$ on the left-hand side of eq. (\ref{NoetherVP2}).
It should be noted however that for a \emph{given} transformation (\ref{inftrans1}), $Q^{\mu}$ is a fixed, infinitesimal
vector field. In particular, for an infinitesimal global symmetry of the form (\ref{infvar1}), addition of the fixed vector
fields $Q^{\mu}_{\mspace{10mu} \rho}$ (as defined in the obvious manner) to the Noether currents (\ref{Noethercurrent})
is dictated by the general solution to Noether's variational problem (\ref{NoetherVP2})
and is \emph{not} equivalent to the addition of so-called improvement terms to the Noether currents.
Terms of the latter sort are of the general form $\partial_{\nu} S^{\nu \mu}_{\mspace{20mu} \rho}$, with $S^{\nu \mu}_{\mspace{20mu} \rho}$ 
anti-symmetric in its first two indices. Addition of an improvement term to any Noether current obviously yields a
conserved current again, which, moreover, gives rise to the same conserved charges (under standard assumptions).
As already noted in the introduction however, by allowing Noether currents to be altered arbitrarily through the addition
of improvement terms, the connection between symmetries and conservation laws is obfuscated.
Although in the specific case of the stress-energy tensor - i.e. the Noether ``current'' associated with constant 
spacetime translations - a method will be seen to exist in subsection \ref{SEMtensor} for defining a unique improvement 
term such that the resulting stress-energy tensor is symmetric, this method is valid only for field configurations 
satisfying their field equations and it moreover does not resolve the general interpretational problem of how to 
precisely connect the resulting conserved current with definite spacetime symmetries. As will also be seen in subsection
\ref{SEMtensor} however, it is possible to write down a general, ``off-shell'' expression involving the very same improved 
stress-energy tensor, via a direct application of theorem \ref{Ntheorem3} for the case of general diffeomorphisms.

\section{Comparison With Standard Formulations}\label{comparison}

\noindent Noether's first theorem is usually formulated as the statement that any global invariance of the action generates
a conserved current. Theorems \ref{Ntheorem1} and \ref{Ntheorem2}, when considered together, endorse this 
standard formulation.
\begin{cor}\label{Ncorollary1}
If the action, $S$, is invariant with respect to a Lie group, $G_r$, of global symmetry transformations of the Minkowski spacetime 
coordinates and/or the fields, the $r$ currents, $\bar{J}^{\mu}_{\mspace{12mu} \rho}$, defined by
\begin{equation}\label{Noethercurrent2}     
\bar{J}^{\mu}_{\mspace{12mu} \rho} \; := \; J^{\mu}_{\mspace{12mu} \rho} \: - \: T^{\mu}_{\mspace{12mu} \nu} X^{\nu}_{\mspace{12mu} \rho}
\end{equation}
with $J^{\mu}_{\mspace{12mu} \rho}$ denoting the usual Noether currents (\ref{Noethercurrent}), are conserved,
if all dynamical fields satisfy their field equations within $\Omega$. 
\end{cor}
\begin{proof}
The divergence relations (\ref{Noetherthm1}) can clearly be expressed as
\begin{equation}
\partial_{\mu} \bar{J}^{\mu}_{\mspace{12mu} \rho} \: + \: E_i (\psi) \gamma^i_{\mspace{12mu} \rho}  \: + \: (\partial^{\mu} T_{\mu \nu}) X^{\nu}_{\mspace{12mu} \rho} \; = \; 0
\end{equation}
while application of theorem \ref{Ntheorem2} to the case of infinitesimal diffeomorphisms, $\delta x^{\mu} = \xi^{\mu}$, gives that 
$\partial^{\mu} T_{\mu \nu} = 0$ if the field equations are satisfied (within the region $\Omega$). Hence, the currents 
(\ref{Noethercurrent2}) are conserved if the field equations are satisfied.
\end{proof}
\noindent An interesting observation is now that the validity of the usual formulation of Noether's first theorem
\emph{as a general statement}, and as characterized above, appears to crucially depend on the validity of the generalized second theorem. 
That is, with the exception of a few special cases, it does not in general appear possible to establish the conservation 
of the tensor $T_{ab}$, the components of which are defined by eq. (\ref{stressenergytensor1}), as a consquence of any 
global symmetry of the action.
Thus, the invariance of $S$ under general diffeomorphisms seems indispensible for the validity of the usual formulation of
Noether's first theorem as a general statement (this view will be further substantiated by the results of subsection
\ref{SEMtensor}). Now, it may appear to the reader that one does not really need the first theorem as a general statement in this 
form. As already remarked before, in the case of spacetime isometries, terms involving the components $T_{\mu \nu}$
are absent and corollary \ref{Ncorollary1} reduces to the standard formulation of the first theorem in terms of the
currents $J^{\mu}_{\mspace{12mu} \rho}$, without any need to appeal to the generalized second theorem.
So if there were good reasons to restrict attention to global symmetry transformations of which the ``external parts'' are
isometries, deviations from Noether's original first theorem would never really arise. This point will be addressed in more
detail shortly. First it may be asked whether the condition that a coordinate transformation be an isometry is in fact a \emph{necessary} 
condition for Noether's original theorem for global symmetries to be valid within the context of Minkowski spacetime. 
This question is partially answered by the following lemma
\begin{lem}\label{Noethersubgroups}
Let $\mbox{O}(3,1)$ denote the familiar Lorentz group and let $\mbox{UT}(4;\real{})$ denote the group of linear 
transformations represented by real four-dimensional, \emph{unitriangular} matrices (i.e. real four-dimensional, 
upper-triangular matrices with unit diagonal entries). Then the condition that is both sufficient and necessary for an 
invertible linear transformation, $A$, to generate a Noether current of the form (\ref{Noethercurrent}), independent of 
the particular form of the action, is that $A \in \mbox{O}(3,1) \times \mbox{UT}(4;\real{})$. 
\end{lem}
\begin{proof}
By identifying the matrix elements, $L^{\mu}_{\mspace{10mu} \nu}$, of an element $L$ of the Lie algebra, 
$\mathfrak{gl}(4;\real{})$, that give rise to an infinitesimal coordinate transformation of the form $\delta x^{\mu} = L^{\mu}_{\mspace{10mu} \nu} x^{\nu}$,
with the infinitesimal parameters $\epsilon_{\rho}$ in the transformations (\ref{infvar1})
according to the scheme $L^0_{\mspace{10mu} 0} \leftrightarrow \epsilon_1 \, , L^0_{\mspace{10mu} 1} \leftrightarrow \epsilon_2 \, , \, \cdots L^1_{\mspace{10mu} 0} \leftrightarrow \epsilon_5 \, , \, \cdots L^3_{\mspace{10mu} 3} \leftrightarrow \epsilon_{16}$
(i.e. by going from left to right, starting from the top downwards, along the entries of $L^{\mu}_{\mspace{10mu} \nu}$), the coefficients $X^{\mu}_{\mspace{12mu} \rho}$ 
in (\ref{infvar1}) are seen to form a $4 \times 16$-matrix, in which the only nonzero elements are the coordinates $x^{\mu}$ and which form the
first row of the first $4 \times 4$-block, the second row of the second $4 \times 4$-block, and so on.
If the original coordinate system is taken to be a global inertial coordinate system, $T_{\mu \nu}$ in eq. (\ref{stressenergytensor1})
is diagonal and the last term in eq. (\ref{Noetherthm1}) is nonzero only for $\rho = 1, 6, 11, 16$, in which case it equals 
respectively $- T_{00} , T_{11} , T_{22} ,  T_{33}$.  
It thus appears that, for a generic $L \in \mathfrak{gl}(4;\real{})$, the last term in eq. (\ref{Noetherthm1}) is absent only
if $T_{\mu \nu}$ vanishes identically, corresponding to a trivial theory. It is clear however, that if the contraction of
$T_{\mu \nu} \, \partial^{\mu} X^{\nu}_{\mspace{12mu} \rho}$ with the infinitesimal parameters $\epsilon^{\rho}$ in eq. 
(\ref{varaction1}) vanishes, $T_{\mu \nu} \, \partial^{\mu} X^{\nu}_{\mspace{12mu} \rho}$ would be absent from eq. (\ref{Noetherthm1})
as well, and under the previous scheme of identification this is found to be the case if
\begin{eqnarray}
T_{\mu \nu} \, \partial^{\mu} X^{\nu}_{\mspace{12mu} \rho} \epsilon^{\rho} & = & - T_{00} L^0_{\mspace{10mu} 0} \: + \: T_{11} L^1_{\mspace{10mu} 1} \: + \: T_{22} L^2_{\mspace{10mu} 2} \: + \:T_{33} L^3_{\mspace{10mu} 3} \nonumber \\
								   	   & = & \mathscr{L} \, \mbox{Tr} L \: + \: 2 \left( \frac{\partial \mathscr{L}}{\partial \eta^{00}} L^0_{\mspace{10mu} 0} \: - \: \frac{\partial \mathscr{L}}{\partial \eta^{11}} L^1_{\mspace{10mu} 1} \: - \: \frac{\partial \mathscr{L}}{\partial \eta^{22}} L^2_{\mspace{10mu} 2} \: - \: \frac{\partial \mathscr{L}}{\partial \eta^{33}} L^3_{\mspace{10mu} 3} \right) \; = \; 0 \nonumber
\end{eqnarray}
In order for this condition to be generic (i.e. independent of the particular form of $\mathscr{L}$), it is necessary that
$L^0_{\mspace{10mu} 0} = L^1_{\mspace{10mu} 1} = L^2_{\mspace{10mu} 2} = L^3_{\mspace{10mu} 3} =0$. This suggests that there
is a remaining $12$-parameter freedom in choosing a generic, global linear coordinate transformation, such that the associated
currents (\ref{Noethercurrent}) are conserved if the field equations are satisfied, by virtue of eq. (\ref{Noetherthm1}).
However, although the set of all real four-dimensional matrices with zero diagonal entries forms a subspace of 
(the matrix algebra corresponding to) $\mathfrak{gl}(4;\real{})$, it is not closed with
respect to the commutator bracket operation and it therefore does not form a Lie subalgebra of $\mathfrak{gl}(4;\real{})$.\\
To obtain a better understanding of the remaining freedom in choosing $L$ such that the last term is absent from eq. (\ref{Noetherthm1}),
it is convenient  to temporarily consider the expression (\ref{varaction1}) for $\delta S$ on the so-called Euclidean section
of complexified Minkowsi spacetime. Mathematically, this makes no difference for any of the discussion that follows eq. (\ref{varaction1}),
except that the ordinary Minkowski time coordinate, $t := x^{0}$ (in units where $c=1$), is replaced by negative imaginary 
time, $- i \tau$ ($\tau > 0$), so that the complexified Minkowski metric, $\tilde{\eta}_{ab}$, is positive definite on the 
Euclidean section.  As a result of the Gram-Schmidt algorithm, the matrix-representative, $A$, of a generic element of $\mbox{GL}(4;\real{})$,
i.e. the group of invertible linear transformations, can be written in the form $A = U O$, with $U$ an 
upper-triangular matrix with positive diagonal entries and $O$ an orthogonal matrix\footnote{In ordinary Minkowski spacetime,
a generalized Gram-Schmidt orthonormalization
procedure still exists, but as a result of the existence of null-vectors, the general product form for $A$ now becomes
$A = E U O$, with $U$ an upper-triangular matrix with nonvanishing diagonal entries, $O$ a pseudo-orthogonal matrix (i.e. a
Lorentz transformation) and $E$ a product of matrices that are obtained from the identity matrix by performing various
types of ``elementary operations'' (such as the interchanging of two columns, the addition of any constant multiple
of one column to another column, etc.) on it.}. Since the set of upper-triangular matrices with positive diagonal entries forms a Lie
group, as does the set of orthogonal matrices, this means that (the matrix-representative of) a generic element of $\mathfrak{gl}(4;\real{})$ 
can be written as the sum of a generic upper-triangular matrix and an anti-symmetric matrix. In particular, (the matrix-representative of) 
an element $L$ of $\mathfrak{gl}(4;\real{})$ with zero diagonal entries can be written as the sum of an upper-triangular matrix with zero
diagonal entries and an anti-symmetric matrix. As the sets of real four-dimensional upper-triangular matrices with zero
diagonal entries and real four-dimensional anti-symmetric matrices both form real six-dimensional Lie algebras, there
are thus two natural Lie subgroups of $\mbox{GL}(4;\real{})$, which give rise to Noether currents of the usual form 
(\ref{Noethercurrent}) via eq. (\ref{Noetherthm1}), namely the familiar Lorentz group, $\mbox{O}(3,1)$, (upon rotating
all expressions back to ordinary Minkowski spacetime - assuming an analytic dependence on the complex spacetime coordinates)
and the group $\mbox{UT}(4;\real{})$.
\end{proof}
\noindent Specific theories may still have larger global invariance groups with corresponding Noether currents of the usual 
form\footnote{For instance, for certain theories, such as Maxwell theory in ordinary Minkowski spacetime, or Klein-Gordon 
theory in a two-dimensional spacetime, the last term in eq. (\ref{Noetherthm1}) is absent also for \emph{conformal} 
isometries, i.e. diffeomorphisms for which $\delta \eta_{ab} = 2 \omega \eta_{ab}$, with $\omega$ some smooth, infinitesimal
function. It is easily verified, moreover, that for such isometries, absence of the last term in eq. (\ref{Noetherthm1}) 
is equivalent to the condition $T^a_{\mspace{12mu}a} = 0$ and anticipating the results of subsection \ref{SEMtensor}, this corresponds 
to the well known fact that a theory is conformally invariant if and only if its stress-energy tensor is traceless.}. 
Generally speaking however, by the above lemma, the only groups of linear transformations that give rise to Noether currents
of the usual form (\ref{Noethercurrent}), are $\mbox{UT}(4;\real{})$ and $\mbox{O}(3,1)$ (and the products of these groups, evidently). 
Currents generated by the former group appear to lack a clear physical interpretation. In addition, by corollary \ref{Ncorollary1}, global 
subgroups of $\mbox{Diff}(\real{4})$ that are not of one of these two types, still generate Noether currents of the non-standard form
(\ref{Noethercurrent2}). As already mentioned, the largest global subgroup of $\mbox{Diff}(\real{4})$ is the group consisting 
of all fractional linear transformations, defined by
\begin{equation}
x^{\mu} \; \longrightarrow \; \frac{A^{\mu}_{\mspace{10mu} \nu} x^{\nu} \: + \: A^{\mu}_{\mspace{10mu} 4}}{A^4_{\mspace{10mu} \nu} x^{\nu} \: + \: A^4_{\mspace{10mu} 4}} \qquad A \in \mbox{GL}(5; \real{}) \qquad \: \mu = 0, 1, 2, 3
\end{equation}
and it is straightforward to check that this group is naturally isomorphic to the projective linear group in five dimensions,
$\mbox{PGL}(5;\real{}) := \mbox{GL}(5;\real{}) / \real{\ast}$ ($\real{\ast}$ denoting the nonzero real numbers), which is
in turn naturally isomorphic to the group $\mbox{SL}(5;\real{})$, consisting of all five-dimensional linear transformations
of unit determinant. As this latter group is clearly $24$-dimensional, it is thus seen that for a generic field theory
in Minkowski spacetime, at least sixteen independent global external symmetries of the action give rise to
Noether currents of the standard form (\ref{Noethercurrent}), while at least four of the remaining eight symmetries give rise 
to Noether currents of the non-standard form (\ref{Noethercurrent2}).
Furthermore, only $10$ of the total number of $24$ global external symmetries appear to admit a clear physical interpretation.\\
This brings us back to the earlier question of whether attention should be restricted to those external symmetries which are
also isometries. In this regard, an objection that could indeed be raised against the entire foregoing line of reasoning, is that it 
completely ignores the fact that, within the context of special relativity, physical laws are not only generally covariant, 
but also obey the principle of special covariance. That is, it is assumed in special relativity that preferred classes of 
motion exist in spacetime, namely the ``inertial'' (or ``nonaccelerating'') motions, which can be used to set up global
\emph{inertial} coordinate systems. 
As the collection of these preferred coordinate systems is in one-to-one correspondence with the group of isometries
of Minkowski spacetime, i.e. the Poincar\'e group, $\mbox{IO}(3,1)$, a good reason thus indeed exists for assigning a special
status to this global subgroup of $\mbox{Diff}(\real{4})$ in the generic, non-theory-specific case\footnote{Because of 
empirical constraints from particle physics it is necessary to add a slight 
refinement to the special relativistic assumption that the inertial motions have a privileged status. Experiments 
indicate that this status should be assigned only to those inertial motions that give rise to (one particular class of) 
global inertial coordinate systems of the same temporal and spatial orientation and these coordinate systems are in bijective correspondence with 
the restricted Poincar\'e group, $\mbox{IO}(3,1)_+^{\uparrow}$ (consisting of those elements $\Lambda \in \mbox{IO}(3,1)$, 
for which $\det \Lambda = 1$ (``$+$'') and $\Lambda^0_{\mspace{8mu}0} = \sqrt{1 + \sum_i (\Lambda^i_{\mspace{8mu}0})^2} > 0$
(``$\uparrow$'')).}.
Yet, although this is true, the formulation of any field theory in Minkowski spacetime in terms of an action principle
is \emph{entirely oblivious} to this fact. As far as the expression (\ref{action1}) for the action is concerned, all global
coordinate transformations are on a par (recall that the coordinate system in (\ref{action1}) is not necessarily inertial 
and that $\mathscr{L}$ is a scalar density). Even if attention is initially restricted to inertial coordinate systems in (\ref{action1}),
there is no reason to distinguish, for instance, between the subgroups $\mbox{O}(3,1)$ and $\mbox{UT}(4;\real{})$; both 
subgroups yield conserved Noether currents of the generic form (\ref{Noethercurrent}), as the previous paragraphs made clear.\\
One possible inference that could be drawn from these observations, is that they demonstrate a weakness in the action-based
approach to field theories.
Another possibility, one that is advocated here, would be to take the general diffeomorphism invariance of any field theory
action seriously - rather than merely as a theoretical artefact, as seems to be the favoured interpretation in many 
contemporary discussions - and see where this leads. 
In fact, a little reflection reveals that even if the actual, physical spacetime of our universe were globally 
isomorphic to Minkowski spacetime, there would still be a conclusive empirical argument against assigning a privileged
status to the inertial motions in spacetime \emph{only}. Indeed, according to the equivalence principle, all material (test) bodies
fall at the exact same rate in a gravitational field located within an effective vacuum and an observer who is studying
the motions of test bodies inside a freely falling laboratory (treated as another test body, thus ignoring possible tidal
effects) and in the absence of non-gravitational, external forces, would therefore validly classify these motions as inertial.
Stated differently, in a hypothetical world in which the equivalence principle holds, but in which spacetime were globally
Minkowskian, with non-gravitational effects described by the physics of special relativity, there would be two distinct 
classes of preferred motions in spacetime, the inertial motions and the ``free-falling'' motions.\\
If one were to place oneself in the position of a physicist trying to come up with a viable description of gravitational
phenomena within the framework of special relativity in the decade following Einstein's 1905 paper, there would thus be
a powerful incentive to attribute physical significance to at least \emph{some} non-isometric, general coordinate transformations.
Only in a strictly hypothetical world, with all gravitational phenomena
excluded, would it be reasonable to attribute ``physical'' significance to the global subgroup $\mbox{IO}(3,1)$ (or its restricted part)
only and to discard all other symmetries in $\mbox{Diff}(\real{4})$ as ``unphysical''.
Of course, within the context of Minkowski spacetime, the inertial motions correspond exactly to the geodesics of spacetime,
whereas the free-fall trajectories lack such a clear geometrical interpretation. So, from the perspective of a physicist
who was trying to fit in gravitational phenomena into the physics of special relativity prior to 1915, this could already
be taken as a suggestion that spacetime is perhaps only ``locally Minkowskian'' and that gravitational effects are somehow to be
attributed to spacetime structure, with the free-fall trajectories corresponding to the geodesics of a more general spacetime
geometry. It is a well known fact, of course, that this was indeed the path followed by Einstein in formulating his 1915 theory of
general relativity.\\
The foregoing remarks do not by any means establish the existence of empirically relevant elements in $\mbox{PGL}(5;\real{}) \backslash \mbox{IO}(3,1)$.
However, these remarks do establish that there does not seem to be any decisive argument \emph{against} the existence of
such elements either - even thus it also seems hard to imagine which elements of $\mbox{PGL}(5;\real{}) \backslash \mbox{IO}(3,1)$
could possibly have empirical relevance in the general case. Yet, this is simply the position that one appears to
be forced into, by adopting the action-based approach to field theories.
In summary, there appears to be no sound, in-principle argument in favour of restricting attention to those external global symmetries
which are also isometries, and the validity of the usual formulation of Noether's first theorem as a general statement
therefore does seem to crucially depend on theorem \ref{Ntheorem2}, i.e. the generalized second theorem (at any rate, it is
very far from clear how such a dependence could be avoided).
Furthermore, regardless of whether some elements of $\mbox{PGL}(5;\real{}) \backslash \mbox{IO}(3,1)$ turn out to have empirical relevance,
the widespread belief that global symmetries of the action \emph{automatically} have empirical relevance is certainly false.
As will be seen in section \ref{applications}, by putting emphasis on the general diffeomorphism invariance of any field 
theory in Minkowski spacetime, as advocated in the present article, it is possible to also obtain several more constructive consequences.\\
Finally, note that by corollary \ref{Ncorollary1}, the ``improper currents'', (\ref{current1}), can alternatively be expressed
as $\mathscr{J}^{\mu}_{\mspace{12mu} \rho} = \bar{J}^{\mu}_{\mspace{12mu} \rho} \: + \: E_i (\psi) ^1 \Psi^{i \mu}_{\mspace{16mu} \rho}$.
It can then be inferred that, in the situation where only the gauge fields transform with nonzero $ ^1 \Psi^{i \mu}_{\mspace{16mu} \rho}$ 
in (\ref{infvar2}) (this covers all known physically relevant cases), the (non-standard) Noether currents are conserved under weaker 
conditions on the field equations than in the context of corollary \ref{Ncorollary1}. 
This does not mean however that the Noether currents are conserved \emph{regardless} of whether the matter 
fields satisfy their field equations, because the Lagrange expressions for the matter and gauge fields are 
not independent by virtue of the identities (\ref{Noetherthm2}).

\section{Applications}\label{applications}
\subsection{The Stress-Energy Tensor}\label{SEMtensor}

\noindent In the case of the group of diffeomorphisms of Minkowski spacetime, application of theorem \ref{Ntheorem3} 
immediately leads to the following result
\begin{cor}\label{SEtensorcor}
The symmetric tensor, $T_{ab}$, the components of which are defined by eq. (\ref{stressenergytensor1}) can be expressed as
\begin{equation}\label{stressenergytensor2}
T_{ab} \; = \; \Theta_{ab} \: + \: \partial^c G_{acb} \: + \: E_i (\psi) ^1 \Psi^{i}_{\mspace{10mu} ab}
\end{equation}
where $\Theta_{ab}$ denotes the canonical stress-energy tensor defined by eq. (\ref{canonicalSEtensor}) and $G_{acb}$, $^1 \Psi^{i}_{\mspace{10mu} ab}$,
are defined by respectively eqs. (\ref{intcurrent1}) and (\ref{infvar2}), for the case of generic diffeomorphisms.
Expression (\ref{stressenergytensor2}) is valid regardless of the satisfaction of any field equations. If the field equations
are satisfied, $T_{ab}$ is conserved, i.e. $\partial^a T_{ab} = 0$, and (under standard assumptions) gives rise to the same
conserved charges as $\Theta_{ab}$. In particular, $T_{ab}$ is a stress-energy tensor for the field(s) $\psi$.
\end{cor}
\begin{proof}
Eq. (\ref{stressenergytensor2}) results from a simple re-ordering of terms in eq. (\ref{utiya2}) in the case of diffeomorphisms
and upon using the expression (\ref{current1}) (with $\bar{X}^{\mu}_{\mspace{12mu} \rho} = \delta^{\mu}_{\mspace{12mu} \rho}$,
$^0 \Psi^i_{\mspace{8mu} \rho} = 0$).
The divergence of $\partial^c G_{acb}$ with respect to its first index vanishes identically because of the symmetry properties
of $G_{acb}$, whereas $\Theta_{ab}$ is conserved if the field equations are satisfied. Alternatively, if the field equations 
are satisfied, conservation of $T_{ab}$ is a trivial consequence of the identities (\ref{Noetherthm2}). 
\end{proof}
\noindent As noted in the introduction, a method for altering the canonical stress-energy tensor into a unique, symmetric
tensor that gives rise to the same conserved charges (under standard assumptions), was already presented a long time ago
by Belinfante\footnote{Cf. \citeANP{Belinfante1} \citeyear{Belinfante1,Belinfante2}. Similar results were obtained by
\citeN{Rosenfeld}. A discussion of Belinfante's method from a modern perspective can be found in \citeN{Weinberg}, while
a generalization of the method for actions involving second-order derivatives of the fields has been given by \citeN{BaCaJa}.}.
It is therefore natural to ask how the expression (\ref{stressenergytensor2})
for the stress-energy tensor relates to the Belinfante tensor. The definition of the latter is based on the two fundamental
observations that (i) the various possible fields form irreducible representations of the Poincar\'e group (or representations
up to sign if spinorial tensor fields are also considered) and that (ii) the Lagrangian density, $\mathscr{L}$, transforms
as a scalar under the Poincar\'e group. In particular, under an infinitesimal Lorentz rotation, $x^{\mu} \rightarrow x^{\mu} + \omega^{\mu}_{\mspace{10mu} \nu}x^{\nu}$,
a generic field $\psi^i$ transforms according to\footnote{It is in principle straightforward to derive the form of the
$(J_{\mu \nu})^i_{\mspace{6mu} j}$ from the fact that, to first order, the variation of a generic tensor field
$T^{\alpha_1 \cdots \alpha_k}_{\mspace{55mu} \beta_1 \cdots \beta_l}$
under an infinitesimal Lorentz rotation $x^{\mu} \rightarrow x^{\mu} + \omega^{\mu}_{\mspace{10mu} \nu}x^{\nu}$ is given by
\begin{equation}
\delta T^{\alpha_1 \cdots \alpha_k}_{\mspace{55mu} \beta_1 \cdots \beta_l} \; = \; \sum_{j=1}^k T^{\alpha_1 \cdots \nu \cdots \alpha_k}_{\mspace{80mu} \beta_1 \cdots \beta_l} \, \omega^{\alpha_j}_{\mspace{18mu} \nu} \: - \: \sum_{j=1}^l T^{\alpha_1 \cdots \alpha_k}_{\mspace{55mu} \beta_1 \cdots \nu \cdots \beta_l} \, \omega^{\nu}_{\mspace{10mu} \beta_j}
\end{equation}}
\begin{equation}\label{Lorentztrans}
\psi^i \rightarrow \psi^i \: + \: \frac{1}{2} \omega^{\mu \nu} (J_{\mu \nu})^i_{\mspace{6mu} j} \psi^j
\end{equation}
where the quantities $(J_{\mu \nu})^i_{\mspace{6mu} j}$ form a representation of the homogeneous Lorentz algebra and 
$(J_{\nu \mu})^i_{\mspace{6mu} j} = - (J_{\mu \nu})^i_{\mspace{6mu} j}$, whereas
\begin{eqnarray}
\delta \mathscr{L}    &    =    &  \frac{\partial \mathscr{L}}{\partial \psi^i} \delta \psi^i \: + \: \frac{\partial \mathscr{L}}{\partial (\partial_{\mu} \psi^i)} \delta \partial_{\mu} \psi^i		\nonumber \\
		      &    =    &  \frac{\partial \mathscr{L}}{\partial \psi^i} \delta \psi^i \: + \: \frac{\partial \mathscr{L}}{\partial (\partial_{\mu} \psi^i)} \partial_{\mu} \delta \psi^i \: - \:  \frac{\partial \mathscr{L}}{\partial (\partial_{\mu} \psi^i)} (\partial_{\nu} \psi^i) \partial_{\mu} \delta x^{\nu} \nonumber \\
		      &    =    &  \frac{1}{2} \omega^{\mu \nu} \left\{ \frac{\partial \mathscr{L}}{\partial \psi^i} (J_{\mu \nu})^i_{\mspace{6mu} j} \psi^j \: + \:  \frac{\partial \mathscr{L}}{\partial (\partial_{\rho} \psi^i)}  (J_{\mu \nu})^i_{\mspace{6mu} j}  \partial_{\rho} \psi^j \: - \; (\Theta_{\mu \nu} - \Theta_{\nu \mu}) \right\} \; = \; 0
\end{eqnarray}
for arbitrary infinitesimal, anti-symmetric matrices $\omega^{\mu \nu}$. Using the field equations, $E_i (\psi) = 0$, one
thus sees that the anti-symmetric part of the canonical stress-energy tensor is given by
\begin{equation}\label{AScanonicalSEtensor}
\Theta_{[\mu \nu]} \; := \; \frac{1}{2}(\Theta_{\mu \nu} - \Theta_{\nu \mu}) \; = \; \frac{1}{2} \partial_{\rho}\left(  \frac{\partial \mathscr{L}}{\partial (\partial_{\rho} \psi^i)}  (J_{\mu \nu})^i_{\mspace{6mu} j} \psi^j \right)
\end{equation}
In order to obtain a symmetric tensor that gives rise to the same conserved charges as $\Theta_{\mu \nu}$, one conventionally
adds the term
\begin{equation}
- \frac{1}{2} \partial_{\rho} \left( \frac{\partial \mathscr{L}}{\partial (\partial^{\mu} \psi^i)}  (J^{\rho}_{\mspace{12mu} \nu})^i_{\mspace{6mu} j} \psi^j  \: + \: \frac{\partial \mathscr{L}}{\partial (\partial^{\nu} \psi^i)}  (J^{\rho}_{\mspace{12mu} \mu})^i_{\mspace{6mu} j} \psi^j \right) 
\end{equation}
to the right-hand side of eq. (\ref{AScanonicalSEtensor}) and then subtracts the resulting tensor from $\Theta_{\mu \nu}$
to obtain the Belinfante tensor, $\mathscr{T}_{\mu \nu}$, i.e.
\begin{equation}\label{Belinfante0}
\mathscr{T}_{\mu \nu} \; := \; \Theta_{\mu \nu} \: - \: \frac{1}{2} \partial^{\rho} \left( \frac{\partial \mathscr{L}}{\partial (\partial^{\rho} \psi^i)}  (J_{\mu \nu})^i_{\mspace{6mu} j} \psi^j \: - \: \frac{\partial \mathscr{L}}{\partial (\partial^{\mu} \psi^i)}  (J_{\rho \nu})^i_{\mspace{6mu} j} \psi^j  \: - \: \frac{\partial \mathscr{L}}{\partial (\partial^{\nu} \psi^i)}  (J_{\rho \mu})^i_{\mspace{6mu} j} \psi^j \right) 
\end{equation}
However, diffeomorphism invariance implies that the improvement term, $\partial^{\rho}S_{\mu \rho \nu}$, defined by the sum of the last three terms of eq. (\ref{Belinfante0})
can be rewritten as $\partial^{\rho}G_{\mu \rho \nu}$, with $G_{\mu \rho \nu}$ defined by eq. (\ref{intcurrent1}).
To see this, note that the Lorentz algebra representatives $(J_{\mu \nu})^i_{\mspace{6mu} j}$, have to satisfy
\begin{equation}\label{Lorentzreps}
(J_{\mu \nu})^i_{\mspace{6mu} j} \psi^j \; = \;  ^1 \Psi^i_{\mspace{10mu} \nu \mu} \: - \: ^1 \Psi^i_{\mspace{10mu} \mu \nu}
\end{equation}
with the coefficients $^1 \Psi^i_{\mspace{10mu} \nu \mu}$ defined by eq. (\ref{infvar2}) (i.e. for a generic coordinate
transformation, $x^{\mu} \rightarrow x^{\mu} + \xi^{\mu}$, $^0 \Psi^i_{\mspace{8mu} \rho} = 0$, the coefficients
$^1 \Psi^{i \mu}_{\mspace{16mu} \rho}$ are \emph{fixed} and determined by the tensor transformation law, while for an 
infinitesimal Lorentz rotation, $x^{\mu} \rightarrow x^{\mu} + \omega^{\mu}_{\mspace{10mu} \nu}x^{\nu}$, one has 
$\delta \psi^i = \, ^1 \Psi^i_{\mspace{10mu} \mu \nu} \omega^{\nu \mu} = \frac{1}{2} \omega^{\mu \nu} (J_{\mu \nu})^i_{\mspace{6mu} j} \psi^j$
by eqs. (\ref{infvar2}) and (\ref{Lorentztrans})).
The improvement term in eq. (\ref{Belinfante0}) can therefore be rewritten according to
\begin{eqnarray}
\partial^{\rho}S_{\mu \rho \nu}  & = & - \: \frac{1}{2} \partial^{\rho} \left( G_{\nu \rho \mu} \: - \:  G_{\mu \rho \nu} \: - \: G_{\nu \mu \rho} \: + \: G_{\rho \mu \nu} \: - \: G_{\mu \nu \rho} \: + \: G_{\rho \nu \mu}  \right)	\nonumber \\
			         & = & \partial^{\rho} G_{\mu \rho \nu}
\end{eqnarray}
(where eq. (\ref{utiya3}) was used several times) and the Belinfante tensor thus becomes
\begin{eqnarray}
\mathscr{T}_{\mu \nu} & = & \Theta_{(\mu \nu)} \: + \: \frac{1}{2} \partial^{\rho} \left( G_{\mu \rho \nu} \: + \: G_{\nu \rho \mu}  \right)  \nonumber \\
		      & = & \Theta_{\mu \nu} \: + \: \partial^{\rho} G_{\mu \rho \nu}							      \label{Belinfante}
\end{eqnarray}
with $\Theta_{(\mu \nu)} \; := \; \frac{1}{2}(\Theta_{\mu \nu} + \Theta_{\nu \mu})$.
\begin{thm}\label{SEtensorthm}
The tensor $T_{ab}$ defined by eq. (\ref{stressenergytensor1}) is the ``specialization'' of the general relativistic 
stress-energy tensor, as defined by the functional derivative of the matter action with respect to the metric. Conversely, the tensor
defined by (\ref{stressenergytensor1}) naturally generalizes to a general relativistic stress-energy
tensor upon employing the minimal substitution rule. When the field equations are satisfied, $T_{ab}$ equals the Belinfante 
tensor (\ref{Belinfante0})
\end{thm}
\begin{proof}
In a general spacetime $(M , g_{ab})$, the stress-energy tensor $T_{ab}$ for any non-gravitational (i.e. ``matter'') field,
or collection of such fields, $\psi^i$, described by an action $S$ is defined by
\begin{equation}\label{GRSEtensor}
T_{ab} \; := \; - \frac{1}{\sqrt{-g}} \frac{\delta S}{\delta g^{ab}}
\end{equation}
(up to an overall proportionality constant). According to the discussion in section \ref{Noetherthms}, $S$ is of the 
general form (\ref{actiongeneric}), with $\mathscr{F}$ a true scalar quantity with a local dependence on $\psi^i$, 
$\nabla_{a_1} \psi^i$, $\cdots$, $\nabla_{a_1} \cdots \nabla_{a_k} \psi^i$, for $k$ finite. As before, it is assumed for convenience
that $k=1$. Thus, for a smooth, one-parameter family of variations of the metric field, one finds
\begin{eqnarray}
\delta S     & = & \int_M \! \{ (\delta \sqrt{-g}) \mathscr{F} \: + \: \sqrt{-g} \delta \mathscr{F} \}     \nonumber   \\
	     & = & - \frac{1}{2} \int_M \! \sqrt{-g} \left\{ (\mathscr{F} g_{ab} \: - \: 2 \frac{\partial \mathscr{F}}{\partial g^{ab}} ) \delta g^{ab} \: - \:  2 \frac{\partial \mathscr{F}}{\partial \nabla_a \psi^i} \delta \nabla_a \psi^i \right\} \label{varaction3}
\end{eqnarray}
where the last term arises from the variation of the derivative operator $\nabla_a$, induced by the variation of the metric.
Modulo a boundary term, this last term can also be written in the form\footnote{See e.g. \citeN{HawEll} for how to achieve 
this, in principle; also note that for simple theories, such as Klein-Gordon theory or Maxwell theory, this last term vanishes.}
$(\cdots)_{ab} \delta g^{ab}$, but for present purposes it is sufficient to note that this term vanishes in the case of Minkowski spacetime
and that, in that case, the definition (\ref{GRSEtensor}) precisely leads back to the definition (\ref{stressenergytensor1}),
up to a simple overall numerical factor and with the proviso that the functional derivative, $\delta S / \delta \eta^{ab}$,
is to be understood symbolically in the flat spacetime case (this is because $\delta S / \delta \eta^{ab}$ in this case is
only defined for directions of variation toward other flat metrics). 
Conversely, the fact that the stress-energy tensor, $T_{ab}$, in Minkowski spacetime, as defined by eq. (\ref{stressenergytensor1}),
naturally generalizes to the definition (\ref{GRSEtensor}) for generic spacetimes upon employing the ``minimal substitution''
rule, $\eta_{ab} \rightarrow g_{ab}$, $\partial_a \rightarrow \nabla_{a}$, easily follows from the observation that $T_{ab}$
can be expressed symbolically as $1 / \sqrt{- \eta} \: \delta S / \delta \eta^{ab}$, modulo an overall proportionality constant (i.e. if
$\eta_{\lambda} := \phi^{\ast}_{\lambda} \eta$ denotes a one-parameter family of field variations obtained by acting on 
$\eta$ with a one-parameter family of diffeomorphisms, $\phi_{\lambda}$, $\lambda \in \real{}$, $\phi_0 := \mathbbm{1}$, it is easily checked that 
$\int_{\real{4}} \! ( \partial \mathscr{L} / \partial \eta^{ab}) \, \delta \eta^{ab} = - 1/2 \int_{\real{4}} \! \sqrt{- \eta} \, T_{ab} \, \delta \eta^{ab}$
- where $\delta \eta^{ab} := (d \eta_{\lambda}^{ab} / d \lambda)_{\lambda = 0}$ and where the global coordinate system 
implicit in the integral is now not necessarily inertial).
\end{proof}
\noindent From the definition (\ref{Belinfante0}) of the Belinfante tensor, it is not at all clear whether this tensor generalizes
to a stress-energy tensor in a curved spacetime that agrees with the definition (\ref{GRSEtensor}). Although it is straightforward
to generalize both the canonical stress-energy tensor and the improvement term in eq. (\ref{Belinfante0}) to curved spacetime
expressions, $\Theta_{ab}$ and $\nabla^c S_{acb}$, respectively, by using the minimal substitution rule, 
$\eta_{ab} \rightarrow g_{ab}$, $\partial_a \rightarrow \nabla_{a}$, Noether's theorem is no longer available to guarantee
the conservation of $\Theta_{ab}$, in general, while the divergence of $\nabla^c S_{acb}$ no longer vanishes identically
because of the spacetime curvature. By contrast, conservation of the stress-energy tensor (\ref{GRSEtensor}) 
is easily shown to be a consequence of diffeomorphism invariance, if the matter fields satisfy their field equations.
It has been known for some time already that the expression for the generalized Belinfante tensor agrees with the definition
(\ref{GRSEtensor}), but the standard proof of this equality turns out to be quite involved\footnote{See \citeN{Kuchar}.
For a more recent derivation of the relation between the tensors $\Theta_{ab}$ and $T_{ab}$, based upon a systematic treatment
of symmetries within the framework of classical theories of gravity formulated in spacetimes of arbitrary  dimensions, see
the Appendix of \citeN{IyerWald}. It should be noted however that, in seeming contrast to the latter treatment, for the
off-shell relation (\ref{stressenergytensor2}) between $\Theta_{ab}$ and $T_{ab}$ to reduce to the on-shell, physical relation 
(\ref{Belinfante}), it is necessary only to assume that matter fields transforming with nonzero $^1 \Psi^{i}_{\mspace{10mu} ab}$ 
satisfy their field equations. For further comparison with \citeANP{IyerWald}'s work, see section \ref{discussion}.}.
Here, the fact that the two tensors are equal has been re-derived as an almost trivial consequence of the 
generalized Noether theorem \ref{Ntheorem3}, as applied to general diffeomorphisms.

\subsection{An Alternative Route to General Relativity}\label{gravity}

\noindent Suppose one were given the task of formulating a field theory for gravity consistent with relativistic principles,
given only the state of knowledge in physics of 1905, but given also, somehow, the entire present state of knowledge in
mathematics (thus including in particular modern-day tensor calculus and the modern theory of group representations). Obviously,
the most straightforward first step would be to attempt to reconcile Newtonian gravity with special relativity, by
modifying the Poisson equation for a gravitational potential, $\phi$, into a relativistic wave equation of the generic
form
\begin{equation}\label{relPoisson}
\Box  \, \phi \; = \; 4 \pi \, \rho
\end{equation}
with $\rho$ representing the mass density of matter present (and where $\Box := \partial_a \partial^a$, as usual, and units
are used in which $G = c =1$).
However, it can be shown that, under quite general conditions, it is impossible to obtain eq. (\ref{relPoisson}) as the leading
order contribution from an equation of motion for a massless field of one of the usual spin types within the context of
special relativity\footnote{See e.g. \citeN{MTW}.}.
Given this, together with the remarks on the equivalence principle near the end of section \ref{comparison},
it then seems most natural to attempt to identify the gravitational field with the dynamical metric tensor, $g_{ab}$, of a
more generic spacetime geometry. In relation to the foregoing discussion, this means the following.
As already remarked before, the appearance of the non-standard term in eqs. (\ref{NoetherVP1}), (\ref{NoetherVP2}), (or, 
alternatively, the appearance of the stress-energy tensor, $T_{ab}$, in the generalized Noether theorems (\ref{Ntheorem1})-(\ref{Ntheorem3})),
can be traced to the fact that the Minkowski metric is not a dynamical field. For the derivative term $\partial \mathscr{L} / \partial \eta^{ab}$
this should be obvious (recalling note \ref{notationremark}), while for the other term it follows from the fact that for an arbitrary,
infinitesimal coordinate transformation, $x^{\mu} \rightarrow x^{\mu} + \delta x^{\mu}$, $\delta \sqrt{- \eta}$ equals
$- \sqrt{- \eta} \, \partial_{\mu} \delta x^{\mu}$ (in fact, these remarks apply equally well to a theory involving a
non-dynamical, curved metric $g_{ab}^{\mbox{\scriptsize stat}}$).
In the case of a dynamical metric field, $g_{ab}$, this means that the same steps that led to lemma \ref{NoetherVPlem} 
now lead to the usual solution to the variational problem, with the non-standard term in eq. (\ref{NoetherVP2}) absorbed 
into the generic sum term, $- E_i \bar{\delta} \psi^i$, involving all the Lagrange expressions for the various dynamical 
fields (and where all expressions are now understood to be valid only locally).
Of course, with $g_{ab}$ dynamic, it is necessary to add kinetic terms for it to the original Lagrangian density for the 
non-gravitational fields. In this regard, it seems reasonable to assume that the total density, $\mathscr{L}$, can be 
obtained by applying the minimal substitution rule, $\eta_{ab} \rightarrow g_{ab}$, $\partial_a \rightarrow \nabla_{a}$, 
to the original density, $\mathscr{L}_{\mbox{\scriptsize M}}$, for the non-gravitational fields, with $\nabla_a$ the 
derivative operator compatible with $g_{ab}$, and then to add to this altered form of $\mathscr{L}_{\mbox{\scriptsize M}}$ 
a density, $\mathscr{L}_{\mbox{\scriptsize G}}$, involving only $g_{ab}$, and a finite number, $k$ say, of its spacetime 
derivatives. Thus, on denoting all non-gravitational fields again by $\psi$, $\mathscr{L}$ is given by
\begin{equation}\label{gravdensity}
\mathscr{L} \; := \; \mathscr{L}_{\mbox{\scriptsize G}}(g , \partial g , \cdots , \partial^k g) \: + \: \mathscr{L}_{\mbox{\scriptsize M}} (g, \psi , \nabla \psi )
\end{equation}
with a corresponding expression for the total action, $S$. Now, the gravitational field equation is of course given by
\begin{equation}\label{gravfieldeq}
\frac{\delta S}{\delta g^{ab}} \; = \; \frac{\delta S_{\mbox{\scriptsize G}}}{\delta g^{ab}} \: + \: \frac{\delta S_{\mbox{\scriptsize M}}}{\delta g^{ab}} \; =: \; \tilde{E}_{ab}(g) \: + \: \frac{\delta S_{\mbox{\scriptsize M}}}{\delta g^{ab}} \; = \; 0
\end{equation}
with $\tilde{E}_{ab}$ denoting the appropriate generalization of the Lagrange expression (\ref{Lagrexpr}) for $\mathscr{L}_{\mbox{\scriptsize G}}$.
But, according to the results of subsection \ref{SEMtensor}, up to a proportionality constant, $(-g)^{-1/2} \delta S_{\mbox{\scriptsize M}} / \delta g^{ab}$
is simply the generalization, $T_{ab}$, of the stress-energy tensor (\ref{stressenergytensor1}), for the non-gravitational 
fields in Minkowksi spacetime, to a generic spacetime $M$ with metric $g_{ab}$ (cf. in particular eq. (\ref{varaction3}) 
and the remarks following it). Thus, the field equation (\ref{gravfieldeq}) for $g_{ab}$ becomes
\begin{equation}
\tilde{E}_{ab}(g) \: + \: \beta \sqrt{-g} T_{ab} \; = \; 0
\end{equation}
$\beta \in \real{}$. On the other hand, it seems reasonable to assume that $\tilde{E}_{ab}$ contains at most second order
derivatives of the metric (under very generic conditions, this guarantees the existence of a well-posed initial value formulation 
for the pure theory\footnote{See e.g. \citeN{Wald}, chapter 10.}).
In that case however, it is known that, in four spacetime dimensions the most general divergence-free, second order expression
for $\tilde{E}_{ab}$ is a linear combination of the Einstein tensor, $G_{ab}$, and the metric itself\footnote{Cf. \citeN{Lovelock}.
Here, the fact that both $T_{ab}$ and $(-g)^{-1/2}\tilde{E}_{ab}$ are divergence-free can easily be shown to be a consequence of 
general diffeomorphism invariance of $S_{\mbox{\scriptsize M}}$ and $S_{\mbox{\scriptsize G}}$, respectively (i.e. without
the need to impose any field equation).}, which can be derived from the familiar Einstein-Hilbert action with a cosmological constant term, i.e.
\begin{equation}
S_{\mbox{\scriptsize G}} \; = \; \int_M \! \sqrt{-g} ( R \: - \: 2 \Lambda)
\end{equation}
with $R$ and $\Lambda$ respectively denoting the scalar curvature and the cosmological constant.
The field equation thus becomes
\begin{equation}
\frac{\delta S}{\delta g^{ab}} \; = \; G_{ab} \: + \: \Lambda g_{ab} \: + \: \beta T_{ab} \; = \; 0
\end{equation}
which reduces to the usual form of Einstein's equation with a cosmological constant upon setting $\beta = - 8 \pi$.
It is worth commenting on the significance of the above derivation.
In many standard treatments Einstein's equation is introduced first on various plausibility grounds (e.g. by noting that
it allows one to recover eq. (\ref{relPoisson}) in the weak-field, slow-motion limit, for $\Lambda$ sufficiently small)
and a variational formulation for general relativity is introduced only at a later stage\footnote{See e.g. \citeN{MTW} or \citeN{Wald}.}.
Consistency then requires that the stress-energy tensor appearing in Einstein's equation be \emph{defined} as the functional
derivative of the action for the non-gravitational fields with respect to the spacetime metric.
By contrast, the above route to Einstein's equation started from a variational formulation for Minkowski spacetime based field
theories. It was observed that the general covariance of such theories leads to generalized Noether theorems and that, as
a result of this, a symmetric, manifestly gauge invariant stress-energy tensor, $T_{ab}$, is obtained, which is given by
\begin{equation}\label{stressenergytensor3}
T_{ab} \; = \; \mathscr{F} \eta_{ab}  - 2 \frac{\partial \mathscr{F}}{\partial \eta^{ab}} \; = \; - \frac{2}{\sqrt{- \eta}} \frac{\partial}{\partial \eta^{ab}} \sqrt{- \eta} \mathscr{F} \; = \; - \frac{2}{\sqrt{- \eta}} \frac{\partial \mathscr{L}}{\partial \eta^{ab}}
\end{equation}
It was then noted that if it is attempted to set up a variational formulation for a field theory within the context of a
more generic spacetime geometry, in which the gravitational field is identified with a dynamical, curved metric of that
geometry, the stress-energy tensor for the non-gravitational fields is given by $(-g)^{-1/2} \delta S / \delta g^{ab}$ (up
to an overall factor) - i.e. the generalization of (\ref{stressenergytensor3}) to a general spacetime geometry.
Thus, in our treatment, the expression (\ref{GRSEtensor}) is not a definition that is required in order to maintain consistency
with Einstein's equation, but is simply a consequence of the form for the ``correct'', physical stress-energy tensor for (non-gravitational)
fields in Minkowski spacetime. With this form for $T_{ab}$, Einstein's equation was then seen to be a consequence of a
variational formulation after imposing a further plausiblity condition on the Lagrange expression, $\tilde{E}_{ab}$, for
the pure theory. Both above routes to general relativity are of course equally valid and the purpose of this subsection
merely was to point out that there exists an alternative route to Einstein's equation in which the generalized Noether 
theorems play a crucial role.

\section{Discussion}\label{discussion}

\noindent Historically, the principle of general covariance has led to a considerable amount of confusion. As mentioned in
section \ref{intro}, this principle may essentially be viewed as expressing the form invariance of the equations of physics
with respect to general coordinate transformations. The significance of the principle in this form as a principle of physics
may be questioned however, since any theory formulated in terms of tensor fields is automatically generally covariant in
this sense\footnote{For an early critique along such lines, see \citeN{Kretschmann}.}. A slightly sharper formulation of
the principle states that the metric of space(time) is the only quantity ``pertaining to space(time)'' that can appear in
any law of physics. This implies for instance that, in the case of a perturbative formulation for a spin-two field, 
$\gamma_{ab}$, initially formulated in ordinary Minkowski spacetime, it should be possible to eliminate all reference 
to the flat background metric in the full nonlinear theory for the spacetime metric, $g_{ab}$, that reduces to $\eta_{ab} + \gamma_{ab}$ 
in the weak-field limit. Formulated this way, the principle can indeed act as a sieve for separating viable from non-viable 
theories of gravity\footnote{See \citeN{Wald1} for a concrete example.}.\\
In the present article, general covariance in the sense of a ``principle'' expressing the diffeomorphism invariance of physical
theories has been used to establish generalized Noether symmetry theorems for field theories in Minkowski spacetime and derivable 
from an action principle. The most notable application of these theorems was seen to consist in new relations between various
different definitions of the stress-energy tensor for any field theory.
It is instructive to compare these findings to related results previously established in the literature\footnote{I thank
an anonymous referee for drawing my attention to these results.}.
In particular, in \citeN{IyerWald}, a general systematic treatment of Noether currents is presented for classical theories 
of gravity formulated in generic spacetimes of arbitrary dimension and for actions which are ``diffeomorphism covariant''.
The technical implementation of this latter condition - not to be confused with the condition of general covariance as 
used in the first sense above - implies the absence of nondynamical fields in the action and it thus seems that there cannot
be any overlap with results established in the present work. Nevertheless, in an Appendix to their paper, \citeANP{IyerWald}
point out that a number of their results continue to hold for theories described by actions that are not diffeomorphism 
covariant and they then go on to derive a relation between the stress-energy tensor as a variational derivative with respect 
to the Minkowski metric and the canonical stress-energy tensor. It should be noted however that their relation (141) in this
regard does not specify a concrete form for the improvement term $\partial_c H^{cab}$, $H^{abc}=H^{[ab]c}$. In addition, the 
notion of symmetries employed by them is non-standard in the sense that a Noether current is associated to each ``infinitesimal 
local symmetry''\footnote{See \citeN{LeeWald} for a precise definition. The term ``Noether current'' has its conventional 
meaning here, in that it is a current associated with a symmetry which is conserved only if the fields satisfy their field equations.}, 
as a result of which the usual global/local distinction prevalent in standard treatments of Noether's theorems (as referred 
to earlier) is largely bypassed. By contrast, the scope of the present work is intrinsically more restricted, as only
classical field theories formulated in ordinary Minkowski spacetime are considered and the usual global/local distinction
in relation to Noether's theorems is largely maintained. Yet, this was seen to lead to \emph{generalized} Noether theorems
(relative to the standard treatments), which in turn also gave rise to a relation between the stress-energy tensor as a 
variational derivative with respect to the Minkowski metric and the canonical stress-energy tensor, but this time with
a concrete form for the improvement term and with an additional off-shell contribution (\ref{stressenergytensor2}). 
It was furthermore shown that the improvement term in this case is nothing but the familiar Belinfante term (\ref{Belinfante}), 
although it was obtained by exploiting general diffeomorphism invariance rather than Poincar\'e invariance.
A slight drawback of the present discussion (especially when compared to the Iyer-Wald treatment), is that it explicitly involves 
use of coordinates at many places. This issue will hopefully be addressed in future work (e.g. by switching to an ``active''
perspective on diffeomorphisms). At any rate, the fact that it is possible to derive the correct form for the stress-energy
tensor from the property that physical theories be generally covariant (as characterized in the first sense above) suggests
that there may be more substance to this property than usually granted.

\bibliographystyle{eigen}
\bibliography{Noetherthms}

\end{document}